%% file: main.tex
\documentclass[nonacm,acmlarge]{acmart}

\AtBeginDocument{%
  }

\usepackage[T1]{fontenc}
\usepackage{graphicx}
\usepackage{caption}
\usepackage{multirow}
\usepackage{amsfonts}
\usepackage{algorithmic}
\usepackage[linesnumbered,lined,boxed,commentsnumbered,ruled,longend, vlined]{algorithm2e}
\usepackage{amsmath}
\usepackage{color}
\usepackage{graphicx}
\usepackage{wrapfig}
\usepackage{flushend}
\usepackage{todonotes}
\usepackage{cleveref}
\usepackage{booktabs}
\usepackage{comment}

\newcommand{{\oB}}{{\overline B}}
\newcommand{{\rh}}{{\widehat r}}
\newcommand{{\Rh}}{{\widehat R}}

\newcommand{\E}{{\mathcal E}}
\newcommand{\T}{{\mathcal T}}

\newcommand{\cO}{{\mathcal O}}

\newcommand{\cb}[1]{{\color{black}#1}}

\begin{document}

\title{A Poly-Log Approximation for Transaction Scheduling in Fog-Cloud Computing and Beyond}



\author{Ramesh Adhikari}
\orcid{0000-0002-8200-9046}
\affiliation{%
  \institution{School of Computer \& Cyber Sciences Augusta University}
  \city{Augusta}
  \state{Georgia}
  \country{USA}
  \postcode{30912}
  }
\email{radhikari@augusta.edu}

\author{Costas Busch}
\orcid{0000-0002-4381-4333}
\affiliation{%
  \institution{School of Computer \& Cyber Sciences Augusta University}
  \city{Augusta}
  \state{Georgia}
  \country{USA}
  \postcode{30912}
}
\email{kbusch@augusta.edu}

\author{Pavan Poudel}
\orcid{0000-0002-0709-9600}
\affiliation{%
  \institution{University of Houston-Clear Lake}
  \city{Houston}
  \state{Texas}
  \country{USA}
  \postcode{30912}
}
\email{poudel@uhcl.edu}









\begin{abstract}
Transaction scheduling is crucial to efficiently allocate shared resources in a conflict-free manner in distributed systems. We investigate the efficient scheduling of transactions in a network of fog-cloud computing model, where transactions and their associated shared objects can move within the network. 
The schedule may require objects to move to transaction nodes, or the transactions to move to the object nodes.
Moreover, the schedule may determine intermediate nodes where both objects and transactions meet.
Our goal is to minimize the total combined cost of the schedule.
We focus on networks of constant doubling dimension, which appear frequently in practice.
We consider a batch problem where an arbitrary set of nodes has transactions that need to be scheduled.
First, we consider a single shared object required by all the transactions and present a scheduling algorithm
that gives an $O(\log n \cdot \log D)$ approximation of the optimal schedule,
where $n$ is the number of nodes and $D$ is the diameter of the network.
Later, we consider transactions accessing multiple shared objects (at most $k$ objects per transaction) and provide a scheduling algorithm that gives an $O(k \cdot \log n \cdot \log D)$ approximation.
We also provide a fully distributed version of the scheduling algorithms where the nodes do not need global knowledge of transactions.
\end{abstract}



\maketitle

\input{introduction}
\input{related-work}
\input{prelimanaries}
\input{basic-algo}
\input{fully-distributed-algo}
\input{conclusion}

\begin{acks}
This paper is supported by NSF grant CNS-2131538.
\end{acks}

\bibliographystyle{ACM-Reference-Format}

\bibliography{references}

\end{document}

%% file: introduction.tex
\section{Introduction}
There are distributed systems that process a large number of concurrent transactions in industry sectors like FinTech, e-commerce, social media, telecommunications, fog-cloud computing, etc. \cite{cao2023transaction,tran2023disco,zhang2023efficient}. A coordination problem arises when multiple transactions attempt to simultaneously read from and write to the same shared objects, such as common accounts. To prevent inconsistencies while accessing shared objects, each transaction should be executed in an atomic way. Traditionally, locks are used to coordinate the actions of transactions and prevent inconsistencies while accessing shared objects \cite{usui2010adaptive}; however, if locks are not handled properly, that leads to a deadlock and priority inversion. To address these issues, we need to efficiently schedule the execution of transactions ensuring that each shared object is accessed by only one transaction at a time.

Consider a distributed system consisting of $n$ processing nodes 
interconnected in a network represented as graph $G$.
A set of transactions $\T$ and a set of shared objects $\cO$ that are required to be accessed by the transactions are initially located at different nodes of the graph $G$.
We consider the case where multiple transactions require to access the shared objects concurrently.
In order for a transaction to execute, it needs to have an exclusive access to all the required objects.
This can be achieved by either moving the objects to the transaction node (data-flow model \cite{herlihy2007distributed,Sharma2014}),
or moving the transaction to the object nodes (control-flow model \cite{PoudelRG24,saad2011snake}).
Recently, Busch {\em et al.}~\cite{busch2023flexible} considered a more flexible transaction execution model called \textit{dual-flow} model where objects and transactions meet at arbitrary nodes of G during the execution.

Busch {\em et al.}~\cite{busch2023flexible} studied transaction scheduling problem in {\em Trees} with the objective of minimizing total communication cost. However, trees are inherently fragile (i.e., they may easily get disconnected on edge removals), lack redundancy, and are inefficient for modeling general networks (e.g., distances may not be preserved and congestion may increase). In contrast, graphs with a constant doubling dimension~\cite{GuptaKL03} provide a compelling alternative, particularly in the context of scalable, fault-tolerant, and real-time systems, by introducing redundancy, robustness, and flexibility needed for real-world high-performance distributed systems. 
Typical examples of graphs with constant doubling dimension are 2-dimensional grids~\cite{abraham2006routing,geo2009distributed} and randomly distributed unit disk graphs~\cite{MulzerW20}. 
Such graphs are used in various distributed and routing systems, such as solving communication and graph problems~\cite{KitamuraKOI21,kuhn2005locality}, resource management~\cite{CeccarelloPPU17,geo2009distributed,srinivasagopalan2011oblivious}, vehicular routing~\cite{JayaprakashS23}, routing and location services in networks~\cite{abraham2006routing,chan2016hierarchical,konjevod2008dynamic}.

In this paper, we studied transaction scheduling in a distributed system modeled as a graph $G$ with a constant doubling dimension, aiming to minimize the total communication cost for executing transactions. Similar to \cite{busch2023flexible}, we also adopted the dual-flow model for transaction execution.
%
%
We assume that moving an object along a unit-length edge in $G$
costs $\alpha$,
while moving a transaction along the unit-length edge costs $\beta < \alpha$.
(The case $\alpha \leq \beta$ corresponds to the well-studied data-flow model~\cite{busch2015impossibility,busch2017fast,busch2022dynamic,Sharma2014,sharma2015load}.)
We consider a synchronous communication model, where time is divided into discrete steps~\cite{herlihy2007distributed}. At each time step, a node may perform one of the following actions: receive objects or transactions from neighboring nodes; execute transactions that have gathered all required objects; or forward objects or transactions to neighboring nodes.

\begin{table*}[t]
\centering
\scriptsize
\centering
    \begin{tabular}{ccccc}
    \toprule
    {\bf Source \hspace{8mm}} & {\bf Scheduling Model} & {\bf Metric} & \multicolumn{2}{c}{\bf Communication cost approximation} \\
    \cmidrule{4-5}
    & & & {\bf Single object} & {\bf Multiple objects}\\
     
    & & & & {\bf (at most $k$)}\\
    \toprule
    
    Busch {\em et al.}~\cite{busch2023flexible}  & Centralized & Tree & $O(1)$ & O(k)\\

    \hline
    {\bf This paper \hspace{3mm}} & Centralized and & Constant doubling \hspace{6mm}& $O(\log n \cdot \log D)$ & $O(k\cdot \log n \cdot \log D)$\\

    & distributed & dimension graph & & \\
    
    \bottomrule
    \end{tabular}
\caption{Comparison of our proposed scheduler with the most related work~\cite{busch2023flexible}, where $n$ is the total number of nodes, $D$ is the diameter of graph $G$, and $k$ is the maximum number of objects accessed by each transaction. }
\label{tbl:contribution-summary}
\end{table*}
\paragraph{\bf Contributions.} 
The primary goal of this paper is to provide an efficient schedule for the execution of transactions in a network $G$ accessing shared objects.
We consider $G$ as a graph of constant doubling dimension.
Table~\ref{tbl:contribution-summary} compares our results with the most related work~\cite{busch2023flexible}. We provide the following contributions:
\begin{itemize}
    \item Assuming each node has global knowledge of all transactions and each transaction accesses a single shared object, we propose a global-aware distributed scheduling algorithm with an $O(\log n \cdot \log D)$ approximation in communication cost,
    where $n$ is total number of nodes and $D$ is the network diameter.
    \item For transactions accessing up to $k$ shared objects, we present a global-aware distributed scheduling algorithm with an $O(k \cdot \log n \cdot \log D)$ approximation in communication cost.
    \item When nodes do not have global knowledge of transactions, we introduce a fully distributed scheduling algorithm for the single-object case and explain how it can be extended to handle multiple shared objects.
\end{itemize}

\paragraph{\bf Techniques.}
\label{subsec:technique}
When there is a single shared object $o$ to be accessed by all the transactions in $\T$, our scheduling algorithm consists of two main steps: 
(1) Find a set of nodes $S_f$ in $G$ where transactions and the object $o$ can meet together with minimal cost, and 
(2) Make the object $o$ visit each node in $S_f$ and execute the transactions in order.
Each node in $S_f$ contains at least one transaction.
When the object visits a node in $S_f$, respective transaction(s)
in the node will execute sequentially.
The overall schedule provides a global order of all transactions in $\T$.
We show that the schedule of our algorithm
has $O(A \log D )$ approximation,
where $A$ be the approximation of the optimal TSP tour for the nodes in $S_f$. 

To calculate the set $S_f$, we use a hierarchical sparse partition $H$ of $G$,
as in~\cite{Jia2005}.
Then the object follows an approximate TSP tour for visiting the nodes of $S_f$.
We can use any of the two ways to calculate the tour.
\begin{itemize}
\item{\em Universal TSP}:
As outlined in \cite{Jia2005},
the hierarchy $H$ can be used to provide a global order of all the nodes in $G$
which is called a {\em universal TSP tour}.
For any subset of $G$, the respective order of nodes in $H$ gives 
an $A=O(\log n)$ approximation of the optimal TSP tour.
When considering $S_f$, using the respective TSP tour,
we obtain the $O(\log n \cdot \log D)$ approximation.

\item{\em MST}:
Alternatively,
we can first calculate a minimum weight spanning tree (MST) with the nodes of $S_f$,
which can be used to approximate the tour with $A = 2$.
Thus, this approach gives an overall approximation of $O(\log D)$ for our proposed algorithm.
\end{itemize}

We would like to note that although the MST approach gives a better approximation,
the distributed version of the approach could require more messages.
In particular, assuming that the hierarchical partition $H$ is given,
the number of messages to compute the tour with $H$ is $O(n \log D)$.
While using the MST approach for the tour, it involves $O(n^2)$ messages.

\begin{wrapfigure}{r}{0.64\textwidth}
  \begin{center}

    \includegraphics[width=0.63\textwidth]{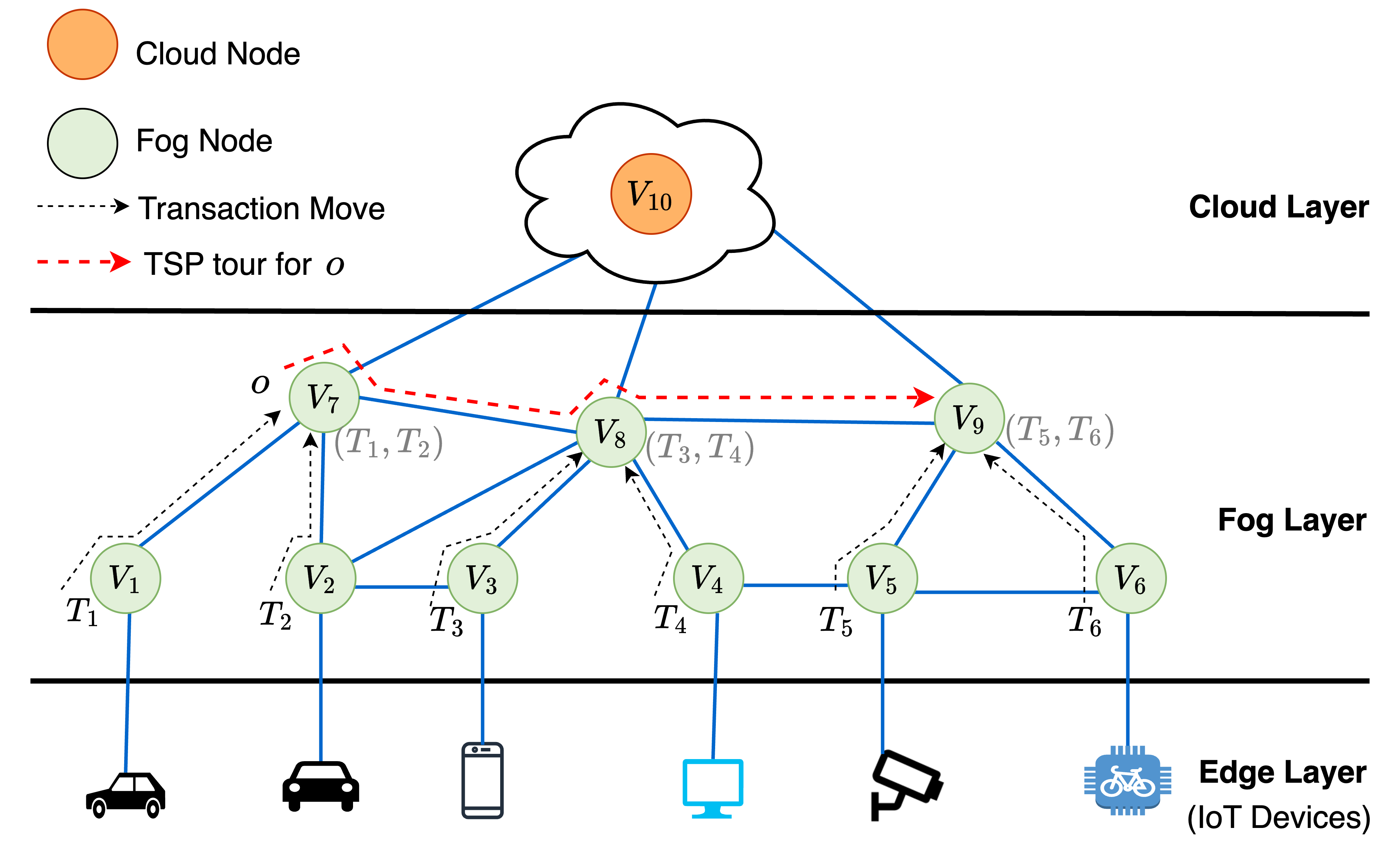}
  \end{center}

  \caption{Illustration of fog-cloud computing model.}

  \label{fig:fog-cloud-computing-model}
\end{wrapfigure}
When transactions are required to access multiple shared objects, we extend the algorithm for a single shared object. First, a set of nodes $S$ in $G$ is calculated with respect to each individual object and the transactions requiring that object where the transactions and the object can meet together with minimal cost. Next, for each transaction $T$, a common node $s'$ is found where all the required objects for $T$ will gather with minimal additional cost. $s'$ is added to the final set of nodes $S_f$ and $T$ will also move to $s'$ for the execution. Following a TSP tour, objects will move to their respective nodes in $S_f$, and the transactions in each node execute sequentially when all the required objects gather at that node.



\paragraph{\bf Applications.}
Our scheduling technique is applicable in many distributed applications such as transactional memory (TM)~\cite{Herlihy1993,shavit1995software}, IoT and fog-cloud computing~\cite{GoudarziPB23,NarmanHAS17,nikoui2020cost,tran2023disco}, financial systems~\cite{GramoliLTZ24}, and Big Data~\cite{KANG2022129}. In the following, we discuss the applicability in fog-cloud computing in detail.


{\em Fog-cloud computing:} IoT applications often rely on centralized cloud computing nodes. However, there exists a high communication delay between cloud and IoT devices \cite{khiat2024genetic,nikoui2020cost}. Distributed fog computing addresses this challenge by providing resources at the edge of the network, closer to end devices that reduce network delay \cite{mokni2023multi,nikoui2020cost}.
In fog-cloud computing model, transactions are generated from user IoT devices and sent to the fog layer.
If the fog node has sufficient resources and the object it requires to execute, then the transaction is scheduled and executed in the fog layer; otherwise, it is routed to a cloud server for execution \cite{khiat2024genetic,mokni2023multi,nikoui2020cost}. Figure~\ref{fig:fog-cloud-computing-model} shows a simple example of the fog-cloud computing. 

As a practical example, consider a vehicle-tracking scenario where IoT devices, like traffic cameras, transmit vehicle data as transactions to the fog layer. Each transaction creates, reads, or updates a vehicle-related object. Due to vehicle mobility, the object may be distant from the transaction node, leading to communication overhead. To reduce this cost, the object, the transaction, or both can be moved. Our proposed algorithm efficiently enables such movements for better transaction scheduling and execution.

\paragraph{\bf Paper Organization:}
The rest of the paper is organized as follows: We discuss related work in Section \ref{sec:related-work} and model and preliminaries in Section \ref{sec:prelimanaries}. We present global-aware distributed scheduling algorithms in Section \ref{section:basic-scheduler}, followed by fully-distributed algorithms in Section~\ref{sec:distributed-scheduler}. We conclude our paper in Section~\ref{sec:conclusion}. Some of the pseudocodes, proofs, and other details are omitted due to space constraints.

%% file: related-work.tex
\section{Related Work}
\label{sec:related-work}
Several studies have explored the scheduling of transactions within distributed systems considering fog-cloud computing~\cite{nikoui2020cost,nikoui2016providing}. Nikoui et al. \cite{nikoui2016providing} employed a genetic algorithm (GA) to minimize energy consumption in task scheduling for green cloud computing systems. However, their approach assumes a central cloud broker utilizing the GA to allocate tasks across a set of virtual machines. Later, Nikoui et al. \cite{nikoui2020cost} introduced a cost-aware genetic-based task scheduling algorithm for fog-cloud environments, also relying on a centralized fog node, termed a fog broker, to handle transaction scheduling. 
Peixoto et al. \cite{peixoto2021hierarchical} proposed a multilevel fog-cloud architecture for transaction scheduling. All of these algorithms depend on a single fog broker, and none of these works considers efficient communication analysis, such as determining when to move the object, the transaction, or both to intermediate nodes, to optimize communication costs and improve the overall efficiency of transaction scheduling.


Extensive research \cite{busch2023stable,busch2017fast,busch2022dynamic,Kim2010,PoudelRG24,PoudelRS21,graphtm} has been done on scheduling transactions in distributed transactional memory systems. In transactional memory, each transaction requires access to specific objects for execution. Coordination between objects and transactions is often necessary for accessing objects and executing transactions. The majority of TM research is based on the data-flow model~\cite{herlihy2007distributed,tilevich2002j}, in which transactions remain static while objects move between nodes to reach the locations where the transactions are held. 
In contrast, several studies adopt the control-flow model~\cite{PoudelRG24,saad2011snake}, where transactions move across nodes to access static objects. 
Lately, the dual-flow model~\cite{busch2023flexible,Hendler2013} has been introduced in which both objects and transactions are moved to some intermediate node to minimize total communication cost. Additionally, there are transaction scheduling techniques in the context of blockchain sharding~\cite{adhikari2023lockless,adhikari2024spaastable}. However, in blockchain sharding, objects are static and only transactions send messages, without a focus on reducing communication costs.



The most closely related work is~\cite{busch2023flexible}, where the authors provide two variants of transaction scheduling algorithms for minimizing communication cost in transactional memory. However, their results are restricted to trees and do not apply to graphs 
with constant doubling dimension that we consider here. Moreover, the algorithms proposed in~\cite{busch2023flexible} 
rely on a centralized model where nodes have access to global information about transactions and objects.
In contrast, this paper also presents a fully distributed version of the algorithms that operate without requiring global knowledge of transactions.
The novelty of our algorithms lies in the use of hierarchical sparse partition $H$ to compute the schedule.

%% file: prelimanaries.tex
\section{Technical Preliminaries and Model}
\label{sec:prelimanaries}
We model the network as a connected weighted graph $G = (V, E, w)$. 
The $n$ vertices in the set $V$ represent the processing nodes that may hold object(s) and transaction(s). Communication links between nodes are represented by edges in the set $E \subseteq V \times V$, and each edge is associated with a weight assigned by the function $w : E \rightarrow \mathbb{R}^+$ to denote the 
distance between the two nodes. The minimum distance between a pair of nodes is $1$. 
A {\em path} $p$ in $G$ is a sequence of nodes with respective edges between adjacent nodes.
There is a path between each pair of nodes and the distance between two nodes varies from $1$ to $D$, where $D$ is the diameter of $G$. 

For a node $v$, the $y$ neighborhood $N_y(v)$, where $y \geq 0$,
is the set of nodes that are at a distance at most $y$ from $v$
(we also include $v$ in $N_y(v)$).
For a set of nodes $S$,
the $y$-neighborhood of $S$ is $N_y(S) = \bigcup_{v \in S} N_y(v)$,
which is all the nodes that are at a distance
at most $y$ from some node in $S$.
The length of a path $p$ in $G$, denoted as $|p|$,
is the sum of the weights of its edges;
if the path is just a single node, then its length is trivially 0.

Similar to the previous work in \cite{busch2023flexible}, we adopt dual-flow model, allowing both shared objects and transactions the flexibility to move between the network nodes.
The cost of moving an object of size $\alpha$ across a unit-weight edge is represented by $\alpha$. Similarly, the cost of moving a transaction across a unit-weight edge is denoted by $\beta$ where $\alpha>\beta$. The scheduling algorithm is responsible for determining the execution schedule $\E$ of the transactions, considering the movements of both objects and transactions in $G$.

We consider graphs with constant doubling dimension ($2^\delta = \Theta(1)$) as described in~\cite{srinivasagopalan2011oblivious}.
In the following definition, a {\em ball of radius $r$} refers to the $N_{r}(v)$-neighborhood of some node $v$.

\begin{definition}[doubling-dimension of graph~\cite{srinivasagopalan2011oblivious}] The doubling dimension of a graph $G$ is the smallest value of $\delta$ such that every ball of radius $r$ in $G$ can be covered by the union of at most $2^\delta$ balls of radius $r/2$. If $\delta$ remains constant, we say $G$ has a constant doubling dimension.
\end{definition}

\subsection{Partition Hierarchy}
\label{section:overlay-tree}
Given a weighted graph $G=(V,E,w)$ with diameter $D \leq n$,
we can build a partition hierarchy $H$ with $O(\log D)$ levels (or layers),
similar to the construction in \cite{Jia2005}.
A $(r,\sigma,I)$-partition
divides $V$ 
into a group of sets $\mathcal{X}=\{X_1,X_2,\dots \}$ such that each set $X_i\in \mathcal{X}$ has diameter at most $r\cdot \sigma$ and for each node $v\in V$, the neighbourhood $N_r(v)$ intersects with at most $I$ sets in the partition;
namely, $|\{X_i : X_i \in \mathcal{X} \wedge X_i \cap N_r(v) \neq \emptyset \}| \leq I$.

$H$ consists of $h+2$ levels, where $h = \lceil \log_\rho D \rceil$,
where $\rho=4\sigma$.
Each level $l \geq 0$ of $H$, is a 
$(r_l,\sigma,I)$-partition $\mathcal{P}_l$ 
computed with $r_l=\min(D,\rho^l)$. 
In each set $X_i$ in $\mathcal{P}_l$, an arbitrary node is designated as leader node $leader(X_i)$.
Let ${\mathcal P}_{-1}$ be the trivial partition where each node 
of $G$ is a cluster by itself.
The maximum level is $h$,
and at that level, the whole graph is a single cluster.
For a leader $\ell$ at level $l < h$,
the parent cluster (and respective parent leader)
is the cluster at layer $\mathcal{P}_{l+1}$
that contains $\ell$.

For graphs with a constant doubling dimension,
the parameters $\rho, \sigma,$ and $I$ are all constant values~\cite{Jia2005}.
For any subset $S'$, let $G'$ be the complete graph consisting of only the nodes in $S'$, such that the edge weight between two nodes in $G'$ is the same as the distance between them in $G$. Then, the order of the nodes for $S'$ in the TSP tour of $G$ gives a $\kappa$-factor approximation for the TSP tour of $S'$ in $G'$ (by connecting the last to the first node in $S'$), where $\kappa = O(\log n)$. 
In \cite{Jia2005},
the authors describe a method to transform the hierarchy $H$ into to a $\kappa$-universal TSP tour for $G$.
Figure~\ref{fig:doubling-dimension-graph} illustrates a partition hierarchy on a constant doubling dimension of graph $G$ with $n=7$ nodes where the left figure shows a partitioning scheme and the right figure shows leader nodes of each cluster and their parents (connected through a virtual link) in the graph $G$. 

\begin{figure}[!t]
    \centering
    \includegraphics[width=0.75\textwidth]{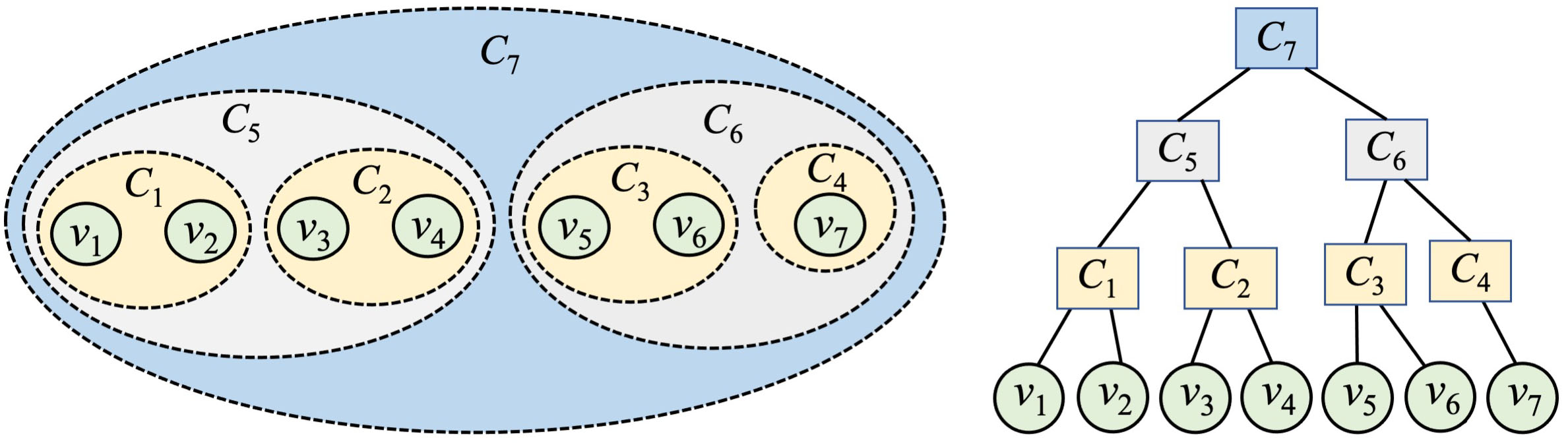}
    \caption{Hierarchical partitioning of graph $G$ with constant doubling dimension.}
    \label{fig:doubling-dimension-graph}
\end{figure}

%% file: basic-algo.tex
\section{Global-Aware Distributed Scheduling}
\label{section:basic-scheduler}
In this section, we provide two basic transaction scheduling algorithms, {\sc GlobalAware\_SingleObj} and {\sc GlobalAware\_MultipleObjs}, for graphs with a constant doubling dimension, assuming that the nodes have global knowledge of transactions and objects they access.

\subsection{Single Object}
\label{sec:single-object}
We consider a single shared object $o$ of size $\alpha>1$ and a set of transactions $\T=\{T_1,T_2,\dots\}$ initially positioned at the nodes of $G$. 
All the transactions in $\T$ require object $o$ for execution. We provide a basic scheduling algorithm denoted as {\sc GlobalAware\_SingleObj} for the execution of transactions accessing object $o$ and its pseudocode is given as Algorithm~\ref{alg:offline-single-obj}. 

A hierarchical partition $H$ of a given graph $G$ can be computed as described in Section~\ref{section:overlay-tree}. 
Now, our task is to find a set of special nodes in $H$ (which we call {\em super-leaders}) to move object $o$ for providing access to the transactions.
The concept of super-leaders introduced here is derived from the notion of supernodes utilized in trees, as outlined in~\cite{busch2023flexible}.
A leader node $v$ of a cluster (set) $X$ of a partition at some level $l \geq 0$ becomes a super-leader if the cluster contains sufficiently large number of transactions such that the cost of moving object $o$ from $v$ to any of the transaction nodes contained in $X$ is less than the total cost of moving all the transactions contained in $X$ to $v$. The idea is that the movement of an object incurs high communication cost compared to the movement of a transaction, thus if some cluster contains trivially less number of transactions, then moving an object to the leader of that cluster results in greater communication cost than transferring the transactions from that cluster to the next-level cluster.
Overall, we attempt to reduce the length of object movement path at each level by computing the super-leaders. 
For each transaction $T_i\in \T$, we find a super-leader at which $T_i$ is executed and we call it a {\em dedicated super-leader} for $T_i$. If a transaction does not have a dedicated super-leader, then it is moved and executed at the node where object $o$ is initially located.

\begin{algorithm}[!t]
\scriptsize
\caption{{\sc GlobalAware\_SingleObj}}
\label{alg:offline-single-obj}
\SetKwInOut{Input}{Input}\SetKwInOut{Output}{Output}
\Input{Graph $G$ with $(r,\sigma,I)$-partition $H$ (with $h+2$ levels), and a set of transactions $\T$}
\BlankLine
$v'\leftarrow$ initial home node for object $o$\; 
$\alpha,\beta\leftarrow$ cost of moving object~$o$ and a transaction over a unit weight edge of $G$, respectively;\\

Initialize $S\leftarrow\emptyset$ and $\gamma \leftarrow\lceil\frac{\alpha}{\beta}\rceil$ \;
\SetKwBlock{DoInRound}{\normalfont {{\bf for each} {\em level $l$ from $0$ to $h+1$ in $H$ with partition $\mathcal{P}_l$} {\bf do}}}{}
\DoInRound{
    \SetKwBlock{DoInRound}{\normalfont {{\bf for each} {\em cluster $X \in \mathcal{P}_l$} {\bf do}}}{}
    \DoInRound{
       $T(X)\leftarrow$ set of transactions contained by the nodes in $X$ that have not been assigned a dedicated super-leader yet\;
        \If{ $|T(X)|\geq \cb{2\gamma}$}
        {
            Add $leader(X)$ to the set of super-leaders $S$\;
            Assign $leader(X)$ as the dedicated super-leader for all the transactions in $T(X)$;
        }
}
}

{\bf for each} $T_i\in\T$, {\bf if} $T_i$ doesn't have a dedicated super-leader, {\bf then} move $T_i$ to $v'$\;

$S_f \subseteq S\leftarrow$ final set of dedicated super-leaders;\\

\SetKwBlock{DoInRound}{\normalfont {{\bf for each} {\em level $l\geq 0$ in $H$ with partition $\mathcal{P}_l$} {\bf do}}}{}
\DoInRound{
   
    Let ${\widehat L}_l$ be the set of dedicated super-leaders at level $l$\;
    ${\widehat T}_l \leftarrow$ set of transactions contained by the dedicated super-leaders at level $l$\;
    \If{$|{\widehat T}_l|<8I\alpha$} {
    Move all the transactions in ${\widehat T}_l$ to $v'$\;
    Update $S_f \leftarrow S_f\setminus {\widehat L}_l$;
    }
}

Move each transaction $T_i\in \T$ to its corresponding dedicated super-leader\;
Calculate TSP tour on the set of dedicated super-leaders $S_f \cup \{ v' \}$\;
Object $o$ traverses the ordered nodes of the 
TSP tour with transactions at the respective node being executed\;
\end{algorithm}

Let $v'$ be the initial home node for object $o$ in $H$.
At each level $l\geq0$ in $H$, if a cluster (set) $X$ of the partition $\mathcal{P}_l$ at level $l$ contains at least $2\gamma$ (where $\gamma=\lceil\frac{\alpha}{\beta}\rceil$) transactions that are not assigned a dedicated super-leader yet, then $leader(X)$ becomes a super-leader and each transaction $T_i \in X$ is assigned $leader(X)$ as a dedicated super-leader. Let $S$ be the set of such super-leaders. 
Each super-leader $s\in S$ also maintains the set of transactions (nodes) ${\widehat T}(s)$ for which $s$ is a dedicated super-leader.
If some transactions do not have a dedicated super-leader yet, then those transactions are moved directly to $v'$ for the execution.

Next, we prune the dedicated super-leaders. 
Let $S_f\subseteq S$ be the final set of dedicated super-leaders. Initially, each dedicated super-leader $s \in S$, where $|{\widehat T}(s)|\geq 2\gamma$, is added to $S_f$.
At each level $l\geq 0$ in $H$, if the sum of total number of transactions contained by the dedicated super-leaders at level $l$ is less than $8I\alpha$, then those transactions will be relocated to the object node for the execution
and all the dedicated super-leaders at level $l$ are removed from $S_f$. This threshold is needed in order to minimize the total communication cost as shown later in the proof of Lemma~\ref{lemma:basic-lower-bound} that if the number of transactions is less than $8I\alpha$, then moving transactions to the node where the object is located is cheaper than moving the object to the transactions node.

In the final stage, each transaction is sent to its assigned dedicated super-leader in the set $S_f$. The goal is to move object $o$ to each super-leader in $S_f$ and execute the corresponding transactions there.
This is achieved by computing a TSP tour starting from $v'$ that visits all nodes in $S_f$.
To compute the tour, we can use the Universal TSP approach based on the hierarchy $H$ as in \cite{Jia2005}.

\begin{wrapfigure}{r}{0.56\textwidth}
\centering
\includegraphics[width=0.56\textwidth]{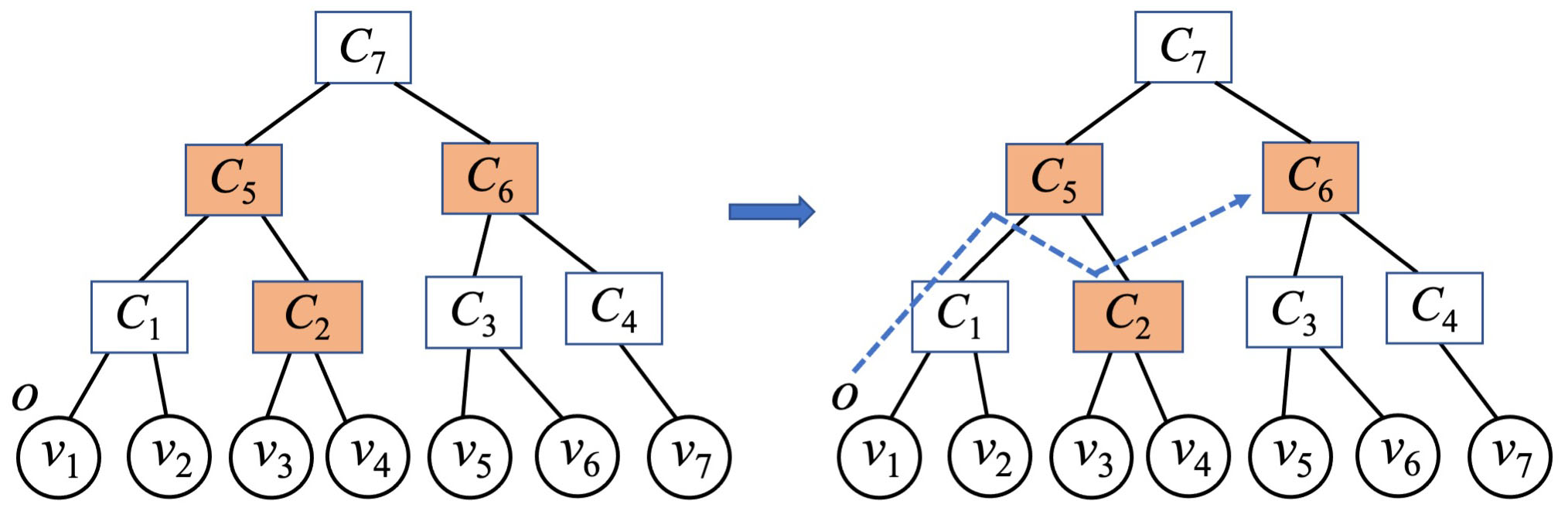}

\caption{Illustration of dedicated super-leaders and the TSP tour for moving object along the dedicated super-leaders in $H$ by Algorithm \ref{alg:offline-single-obj}.}

\label{fig:leader-and-superleader}
\end{wrapfigure}
Figure \ref{fig:leader-and-superleader} illustrates an execution of Algorithm \ref{alg:offline-single-obj}.
The figure in the left shows the set of dedicated super-leaders highlighted in orange color.
Assuming that the object $o$ is initially positioned at node $v_1$, the dashed line on the right figure traces the TSP tour for $o$ moving from $v_1$ to the dedicated super-leaders executing the transactions in each node.

\paragraph{\bf Analysis.}
Without loss of generality, assume that $\alpha > 1$,
and $\beta = 1$.
Hence, $\gamma = \alpha / \beta = \alpha$,
and the criterion of a leader becoming a super-leader is 
$2\gamma = 2 \alpha$. 

Each involved node in $V$ holds at most one transaction.
Let ${\widehat T}_i \subseteq \T$ denotes the set of transactions whose dedicated super-leader is at level $i$.
Let ${\widehat L}_i$ denotes the set of dedicated leaders at level $i$
(i.e., each $v \in {\widehat L}_i$ has a transaction that picked it as dedicated).
Let $Tour^*({\widehat L}_i)$ denotes the optimal length TSP tour for the nodes ${\widehat L}_i$,
which also includes the original position of the object $o$.
Let $C^*$ be the optimal cost of executing all the transactions.

For a set of $Z$ transactions, let $C_{\text{direct}}(Z)$ be the cost of a schedule that sends all $Z$ transactions to the object's node (object does not move in this schedule). 

\begin{lemma}
\label{lemma:direct}
   For $Z > 0$ transactions, 
   $C_{\text{direct}}(Z) \leq 4 (|Z|/\alpha + \log_2 D) C^*$.
\end{lemma}

\begin{proof}
Let $u$ be the node with the object $o$.
Let $Z_{i} \subseteq Z$, $i \geq 0$, be the set of transactions whose distance from $u$ is in the range $[2^i, 2^{i+1})$.
Let $C(Z_i)$ be the cost of moving the transactions $Z_i$ to $u$.
We have $C(Z_i) < |Z_i|2^{i+1}$.
    
In the best case scenario, all $Z_{i}$ transactions are gathered at a single node $v$ at distance $2^{i}$ from $u$.
Ideally, the optimal schedule would move $o$ in an intermediate node $x$ between $u$ and $v$ (or $x$ may coincide with $u$ or $v$).
Let $d_1$ be the distance between $u$ and $x$ and $d_2$ be the distance between $x$ and $v$, where $d_1 + d_2 = 2^i$, and $d_1, d_2 \geq 0$.
The cost of moving object $o$ to $x$ is $d_1 \alpha$.
The cost of moving the $Z_i$ transactions to $x$ is $d_2 |Z_i|$.
Thus, $C^* = d_1 \alpha + d_2 |Z_i|$.
We examine two cases:

\begin{itemize}
    \item    
    $|Z_i| < \alpha$:
    then $C^* > d_1 |Z_i| + d_2 |Z_i| = 2^i |Z_i|$.
    Hence, $C(Z_i) < |Z_i|2^{i+1} < 2 C^* < 2 (|Z_i|/\alpha + 1)C^*$.

    \item
    $|Z_i| \geq \alpha$:
    then $C^* \geq d_1 \alpha + d_2 \alpha = 2^i \alpha$.
    Hence, $C(Z_i) < |Z_i|2^{i+1} = |Z_i|2^{i+1}\alpha /\alpha$ $\leq 2 |Z_i| C^* / \alpha < 2 (|Z_i|/\alpha + 1)C^*$.
\end{itemize}
    
Hence, $C(Z_i) \leq 2 (|Z_i|/\alpha + 1)C^*$.
Since $0 \leq i \leq \lceil \log_2 D \rceil$, and $\sum_{i=0}^{\lceil \log_2 D \rceil} |Z_i| = |Z|$, we have:

\[C_{\text{direct}}(Z)\leq\sum_{i=0}^{\lceil \log_2 D \rceil} C(Z_i) \leq \sum_{i=0}^{\lceil \log_2 D \rceil} 2 (|Z_i|/\alpha + 1)C^* = 2 C^* \sum_{i=0}^{\lceil \log_2 D \rceil} (|Z_i|/\alpha  + 1)\]
    
$\hspace{11mm}= 2 C^* (|Z| /\alpha + 1 + \lceil \log_2 D \rceil)
\leq 4 (|Z|/\alpha + \log_2 D)C^*.$
\qed
\end{proof}

Let $q$ be the path of the object in an optimal schedule for all the 
transactions in $\T$ which results in the optimal schedule with cost $C^*$.
We can recursively decompose $q$ to a sequence of edge-disjoint 
subpaths $q_1, q_2, \ldots, q_k$, $k \geq 1$,
with parameter $\xi$
as follows. 
Let $q_1$ be the shortest prefix subpath of $q$ 
such that its length $|q_1|$ is at least 
$|q_1| \geq \xi$.
If such $q_1$ does not exist,
then $q_1 = q$.
Let $q'$ be the remaining subpath of $q$
where the first node of $q'$ is the same with the last node of $q_1$.
If $q'$ consists of at least two nodes,
repeat the same process recursively on $q'$.
We have the following properties for optimal path $q$:
\begin{itemize}

\item
The first and last node of $q$ is the first of $q_1$, and the last of $q_k$, respectively.
\item
The first node of $q_j$ is the same 
with the last node of $q_{j-1}$, where $j > 1$.
\item
For $1 \leq j < k$,
each $q_j$ has a length
at least 
$|q_j| \geq \xi$
such that 
the removal of the last edge splits $q_j$
into prefix $q'_j$ and suffix $q''_j$
with $|q'_j| < \xi$.
\item
For the last subpath $q_k$,
if it consists of three or more nodes, then
for its prefix $q'_j$, it holds $|q'_j| < \xi$;
otherwise, $|q_k| > 0$.
\item 
$k \leq |q| / \xi + 1$:
Since each of the first $k-1$ segments 
has length at least $\xi$,
we have $|q| \geq (k-1) \xi$. Therefore,
$k-1 \leq |q|/ \xi$, 
which implies $k \leq |q|/\xi + 1$.
\end{itemize}

\begin{lemma}
\label{lemma:subpath}
For a path $q_j$, $1 \leq j \leq k$,
with $v_1$ and $v_2$ being the first and last nodes of $q_j$,
any $c \geq 0$,
$N_{c}(q_j) \subseteq N_{2c}(\{v_1, v_2\})$.
\end{lemma}
\begin{proof}

This holds immediately if $q_j$ consists only of $v_1$ and $v_2$,
since $N_{c}(q_j) = N_{c}(\{v_1 \cup v_2\}) \subseteq N_{2c}(\{v_1, v_2\})$.
Suppose now that $q_j$ consists of three or more nodes.
Let $v'$ be the node of $q_j$ which is adjacent to $v_2$.
We have that $dist_G(v_1,v') \leq c$.
Thus, all the nodes of $q_j$ between $v_1$ and $v'$ are at distance less
than $\xi$ from $v_1$;
hence, the $\xi$-neighbors of all these nodes
are at a distance less than $2 \xi$ from $v_1$.
Thus, $N_{c}(q_j) \subseteq N_{2c}(v_1) 
\cup N_{c}(v_2) \subseteq N_{2c}(\{v_1, v_2\})$.
\qed
\end{proof}

\begin{lemma}
\label{lemma:basic-lower-bound}
For $|{\widehat T}_i| \geq 8I\alpha$, 
$C^* \geq |{\widehat T}_i| \rho^{i-1} / (16I)$. 
\end{lemma}
The proof of Lemma~\ref{lemma:basic-lower-bound} is omitted due to space constraints.
\qed

\begin{lemma}
\label{lemma:doubling-special}
The number of $\rho^i$-neighborhoods that are required to cover 
a $8 \sigma \rho^{i}$-neighborhood of $G$ is at most $\zeta := 2^{\delta \log(8\sigma)}$.  
\end{lemma}

\begin{proof}

Let $G$ be a graph with doubling dimension $\delta$.
Consider a ball of radius $8\sigma\rho^i$ centered at a vertex $v$ in $G$, denoted as $B(v, 8\sigma\rho^i)$. We want to cover this ball with $\rho^i$-neighborhoods.
From the definition, any ball of radius $2r$ can be covered by at most $2^\delta$ balls of radius $r$.
Here, $r = 4\sigma\rho^i$. Therefore, the ball $B(v, 8\sigma\rho^i)$ can be covered by at most $2^\delta$ balls of radius $4\sigma\rho^i$. Similarly, to cover $4\sigma\rho^i$, we need $2^\delta$ balls of radius $2\sigma\rho^i$ and so on. We can estimate the total number of balls as $\prod_{i=1}^{k} 2^{\delta}=(2^\delta)^k$,
where $\frac{8\sigma\rho^i}{2^k}=\rho^i$. Thus, $k=\log(8\sigma)$.
Hence, $2^{\delta k} = 2^{\delta \log(8\sigma)}$.
\qed
\end{proof}

\begin{lemma}
\label{lemma:L-bound}
${\widehat L}_i \subseteq N_{4\sigma \rho^{i}}(q)$.
\end{lemma}

\begin{proof}

Suppose that there is $v \in {\widehat L}_i$,
such that $v \notin N_{4\sigma \rho^{i}}(q)$.
Let $X$ be the cluster of $v$ at partition $\mathcal{P}_{i}$.
Let $d$ be the distance of $v$ to the closest node in $q$;
note that $d > 4\sigma \rho^{i}$.
The distance between any node in $X$ to its closest node in $q$ 
is at least $d - \sigma \rho^i$.
Let $T(X)$ be the transactions with home node in $X$.
Since $v$ is a super-leader,
$|T(X)| \geq 2\alpha$.
The cost $c_1$ of moving these transactions to $q$ is
$$c_1 
\geq (d - \sigma \rho^i) 2 \alpha 
= (2 d - 2 \sigma \rho^i) \alpha 
> (d + 4\sigma \rho^i - 2 \sigma \rho^i ) \alpha 
= (d + 2 \sigma \rho^i) \alpha \ .$$
On the other hand, 
the cost $c_2$ of gathering the $T(X)$ transactions to 
$v$ and then having the object move from $q$ to $v$ 
is $c_2 \leq 2 \alpha \sigma \rho^i + d \alpha < c_1$.
Therefore, the path $q$ is not optimal; a contradiction.
Thus, ${\widehat L}_i \subseteq N_{4\sigma \rho^{i}}(q)$.
\qed
\end{proof}

Using Lemmas~\ref{lemma:subpath}, \ref{lemma:doubling-special}, and \ref{lemma:L-bound}, we obtain the following two results:

\begin{lemma}
\label{lemma:large-q}
For $|{\widehat T}_i| \geq 4I\alpha$
and $|q| \geq \rho^{i-1}/2$,
$Tour^*({\widehat L}_i) \leq 74 \zeta I \sigma \rho |q|$.
\qed
\end{lemma}

\begin{lemma}
\label{lemma:small-q}
For $|{\widehat T}_i|\geq 4I\alpha$ and $|q| < \rho^{i-1}/2$,
$Tour^*({\widehat L}_i) \leq 36 C^* \rho \zeta \sigma / \alpha$.
\qed
\end{lemma}

Let $R$ be the set of indices in range $0, \ldots, h$ such that for each $i \in R$,
$|{\widehat T}_i|\geq 8I\alpha$ and ${\widehat T} = \bigcup_{i \in R} {\widehat T}_i$.
We have a set of dedicated super-leaders $S_f = \bigcup_{i \in R} {\widehat L}_i$.
Let $Tour^*(S_f)$
denote the cost of the optimal TSP tour
covering all the nodes in $S_f$,
including the object's initial position, and $Tour(S_f)$ be the actual cost of the object tour incurred by Algorithm \ref{alg:offline-single-obj}. 
(Without loss of generality, assume that $Tour^*(S_f) > 0$, since otherwise we would only have the direct move cost for the transactions.) We establish the following three Lemmas (proofs are omitted for brevity).

\begin{lemma}
\label{lemma:tour-bound}
$\alpha \cdot Tour^*(S_f) \leq 74 (h+1) \zeta I \sigma \rho C^*$.
\qed
\end{lemma}


\begin{lemma}
\label{lemma:in-R}
For transactions ${\widehat T}$, 
cost of Algorithm \ref{alg:offline-single-obj}
is $C({\widehat T}) \leq 74 \frac {Tour(S_f)}{Tour^*(S_f)} (h + 1) \zeta I \sigma \rho C^*$.
\qed
\end{lemma}

\begin{lemma}
\label{lemma:not-R}
For transactions $\T \setminus {\widehat T}$, 
cost of Algorithm \ref{alg:offline-single-obj}
is $C_{\text{direct}}(\T \setminus {\widehat T}) \leq 36 (h+2) I C^* + 4 C^* \log_2 D$.
\qed
\end{lemma}


\begin{theorem}
\label{thereom:main}
For executing all transactions in $\T$, the total communication cost of Algorithm \ref{alg:offline-single-obj}
is $C = O \left( C^* \cdot \frac {Tour(S_f)}{Tour^*(S_f)} \cdot \log D \right )$.
\end{theorem}

\begin{proof}

For the cost $C$ of Algorithm \ref{alg:offline-single-obj}
we have,
$C = C({\widehat T}) + C_{\text{direct}}(\T \setminus {\widehat T}) \ .$

\noindent{Thus, from Lemmas \ref{lemma:in-R} and \ref{lemma:not-R}, we get:}
$$C \leq 74 \frac {Tour(S_f)}{Tour^*(S_f)} (h + 1) \zeta I \sigma \rho C^* + 36 (h+2) I C^* + 4 C^* \log_2 D.$$
The result follows since $h = O(\log D)$, and $\sigma, I, \rho, \zeta$ are constants
(see Lemma \ref{lemma:doubling-special} for $\zeta$ being a constant).
\qed
\end{proof}

When using a TSP tour based on the minimum weight spanning tree, we get ${Tour(S_f)} / {Tour^*(S_f)} \leq 2$.
Thus, from Theorem \ref{thereom:main},
we get the following result.
\begin{corollary}
Using a MST based TSP to implement the tour of the object,
the approximation of Algorithm~\ref{alg:offline-single-obj}
is $O(\log D)$.
\qed
\end{corollary}

Using a universal TSP tour based on $H$,
we get ${Tour(S_f)} / {Tour^*(S_f)} = O(\log n)$.
Thus, from Theorem \ref{thereom:main},
we get the following result.

\begin{corollary}
Using a universal TSP to implement the tour of the object,
the approximation of Algorithm~\ref{alg:offline-single-obj}
is $O(\log n \cdot \log D)$.
\qed
\end{corollary}

\subsection{Multiple Objects}
\label{sec:multiple-object}
We consider a set of shared objects $\cO =\{o_1, o_2, \ldots\}$, all initially positioned at some node $v'$ of $G$. Each transaction $T_i \in \T$ accesses at most $k$ objects from $\cO$. The set of objects accessed by transaction $T_i$ is denoted as $objs(T_i ) \subseteq \cO$. We assume that nodes have global knowledge of transactions and provide a scheduling algorithm {\sc GlobalAware\_MultipleObjs} for executing the transactions in $\T$. The pseudocode is provided in Algorithm~\ref{alg:offline-multiple-obj-alg}. 

In Algorithm~\ref{alg:offline-multiple-obj-alg}, the objective is to ensure synchronized access to the objects at minimal cost for executing the transactions.
We accomplish this by first using Algorithm \ref{alg:offline-single-obj} to calculate the dedicated super-leaders for each object individually.
We then combine some of the dedicated super-leaders of the objects in a set of nodes $S_f$.
Finally, we use a single TSP tour over $S_f \cup \{ v'\}$, where the objects move synchronously from node to node along the tour, executing the corresponding transactions at each visited node.
Note that an object visits only the nodes in the TSP tour where it is needed for a transaction, skipping all others.

The set of dedicated super-leaders $S_f$ is determined as follows.
Each transaction $T_i\in \T$ accesses up to $k$ objects, potentially with different dedicated super-leaders for each such object in $objs(T_i)$.
For each $T_i$, we select the nearest super-leader among those assigned to its objects, and designate it as the $T_i$'s dedicated super-leader $s$. 
Transaction $T_i$ moves to $s$, and $s$ is added to $S_f$.
During the TSP traversal of $S_f$, all objects in $objs(T_i)$ will be brought to the leader node $s$, and $T_i$ is executed.
If $T_i$ has no dedicated super-leader (i.e., none of its objects do), then it moves to $v'$ and is executed when $v'$ is visited in the tour.



\begin{algorithm}[!t]
\scriptsize
\caption{\sc GlobalAware\_MultipleObjs}
\label{alg:offline-multiple-obj-alg}
\SetKwInOut{Input}{Input}\SetKwInOut{Output}{Output}
\Input{Graph $G=(V,E, w)$ with $(r,\sigma,I)$-partition $H$, and a set of transactions $\T$}
\BlankLine
$v'\leftarrow$ initial home node for all objects in $\cO$; 
$S_f \gets \emptyset$\;
\BlankLine
\tcp{Find super-leaders w.r.t. each object}
\SetKwBlock{DoInRound}{\normalfont {{\bf for each} {\em object $o_j \in \cO$} {\bf do}}}{}
\DoInRound{
Use Algorithm~\ref{alg:offline-single-obj} to find the dedicated super-leaders for object $o_j$\;
}

\BlankLine
\tcp{Find dedicated super-leader for $T_i$}
\SetKwBlock{DoInRound}{\normalfont {{\bf for each} {\em transaction $T_i \in \T$} {\bf do}}}{}
\DoInRound{
    \If{$T_i$ has at least one dedicated super-leader assigned above for any of the objects in $objs(T_i)$}
    {
    Find the closest super-leader $s$ in $H$ among all the dedicated super-leaders assigned to $T_i$ and mark it as the dedicated super-leader for $T_i$\;
    Move $T_i$ to $s$ and
    add $s$ to $S_f$\; 
    }
    {\bf else}
    move $T_i$ to the initial home node of objects, $v'$\;  
}

\BlankLine
\tcp{Build and traverse TSP tour}


Calculate TSP tour for 
$S_f\cup \{v'\}$\;
Each object $o_j \in \cO$ follows the order of the TSP tour for visiting the nodes where $o_j$ is required to execute respective transactions\;

\end{algorithm}

\vspace{1mm}
\paragraph{\bf Analysis.}
We continue with the analysis of Algorithm~\ref{alg:offline-multiple-obj-alg}.
Consider an object $o_j \in \cO$. 
If a super-leader $s'$ of $o_j$, as returned by Algorithm~\ref{alg:offline-single-obj}, is not used in the set $S_f$ calculated by Algorithm~\ref{alg:offline-multiple-obj-alg},
then another super-leader is being used in $S_f$ which is closer to the transactions that were going to move to $s'$.
Thus, since $\alpha > \beta$, the cost of such a change of super-leader does not incur an additional cost with respect to object $o_j$ (where cost of $o_j$ is as calculated in Theorem \ref{thereom:main}).

Now consider the case where object $o_j$ is accessed by a transaction $T_i$, originally located at some node $u$, and $T_i$ 
does not have a dedicated super-leader for $o_j$ under Algorithm~\ref{alg:offline-single-obj},
but does have a dedicated super-leader $s$ under Algorithm~\ref{alg:offline-multiple-obj-alg}.
If the distance from $u$ to $s$ is no greater than the distance from $u$ to $v'$,
then the cost of accessing $o_j$ remains unaffected (as calculated in Theorem \ref{thereom:main}).

If the distance from $u$ to $s$ is more than the distance from $u$ to $v'$,
then there is additional cost due to the transaction $T_i$ 
traveling a longer distance in Algorithm~\ref{alg:offline-multiple-obj-alg}.
However, since $T_i$ has $s$ as a dedicated super-leader, there must exist some 
object $o_z \in \cO$ which is accessed by $T_i$ for which $s$ is a dedicated super-leader for $T_i$.
The additional cost with respect to object $o_j$ is no more than the cost $C'$ of moving $T_i$ to $s$,
as counted in the analysis of object $o_z$ (in Theorem \ref{thereom:main}).
Since $T_i$ can access at most $k-1$ other objects similar to $o_j$, the additional cost for moving these objects to $s$ is at most $(k-1)C'$.

Hence, using $S_f$ increases the total cost by at most a factor of $k$ compared to the sum of the costs of individual objects. As a result, the approximation of Theorem~\ref{thereom:main}
scales by a factor of $k$, giving the following theorem and corollaries:



\begin{theorem}
\label{thereom:main2}
For executing all transactions in $\T$, 
the total communication cost of Algorithm~\ref{alg:offline-multiple-obj-alg}
is $C = O \left( k \cdot C^* \cdot \frac {Tour(S_f)}{Tour^*(S_f)} \cdot \log D \right )$.
\qed
\end{theorem}

\begin{corollary}
Using a MST based TSP to implement the tour of the object,
the approximation of Algorithm~\ref{alg:offline-multiple-obj-alg}
is $O(k \cdot \log D)$.
\qed
\end{corollary}

\begin{corollary}
Using a universal TSP to implement the tour of the object,
the approximation of Algorithm~\ref{alg:offline-multiple-obj-alg}
is $O(k \cdot\log n \cdot \log D)$.
\qed
\end{corollary}

%% file: fully-distributed-algo.tex
\section{Fully Distributed Scheduling}
\label{sec:distributed-scheduler}
In this Section, we provide fully distributed transaction scheduling algorithms assuming that the nodes in $G$ do not have the global knowledge of transactions. First, we discuss the algorithm in detail for the case where transactions access a single shared object. Later, we provide a high-level idea for extending the algorithm for multiple shared objects.

\subsection{Single Object}
\label{sec:fully-dis-single-object}
As in Section~\ref{sec:single-object}, we consider a single shared object $o$ of size $\alpha > 1$ and a set of transactions $\T=\{T_1,T_2,\dots\}$
initially positioned at arbitrary nodes of $G$. Each transaction $T_i\in\T$ requires access to object $o$ for the execution. We provide a fully distributed scheduling algorithm {\sc FullyDistributed\_SingleObj} for the execution of transactions in $\T$ accessing the object $o$. 
The algorithm operates in three phases:
(1) Cluster leaders share transaction information with their parent leaders and identify super-leaders;
(2) Each transaction selects its dedicated super-leader and moves to that node for execution;
(3) A TSP tour is computed to move the object to each dedicated super-leader for transaction execution.
We now describe each phase in detail.

\vspace{1mm}
{\bf Phase 1:} 
In this phase, transaction information is propagated through the levels of the hierarchy $H$ to identify the super-leaders. 
The process begins with each node simultaneously sharing its transaction data with the leader of the corresponding cluster at level $l=0$ in $H$. 
At each level $l\geq 0$, the leader of each cluster (set) $X$ ($leader(X)$) in the partition $\mathcal{P}_l$ counts the transactions within its cluster. If the count reaches at least $2\gamma$, the cluster leader becomes a super-leader and notifies all child leaders recursively down to the level $0$. It also forwards the transaction count to the parent cluster at level $l+1$, continuing this process up to the root of the hierarchy $H$. At the end of Phase 1, a set $S$ of all possible super-leaders is computed, and each cluster leader has the knowledge of $S$. 

\vspace{1mm}
{\bf Phase 2:} In this phase, each node independently determines the dedicated super-leader for its associated transaction $T_i\in\T$. Let $X$ be the lowest-level cluster in $H$ that contains $T_i$, and
$S$ be the set of super-leaders whose information was received by $leader(X)$ from its parent clusters. Then, $T_i$ selects a dedicated super-leader $s\in S$ such that $s$ is the closest super-leader to $T_i$ among all in $S$.
If $S$ is empty (i.e., $leader(X)$ didn't receive any super-leader information from its parents), then $T_i$ is moved to the node where the object $o$ is initially located. 

The set of dedicated super-leaders is further pruned as follows. At each level $l\geq 0$ in $H$, the algorithm checks the total number of transactions across all dedicated super-leaders. If this total is less than $8I\alpha$, then the associated transactions are moved to the initial home node of object $o$ and the corresponding dedicated super-leaders are unmarked. 
To compute the total at level $l$, a reference leader node $s_z$ is chosen at the level $l$, and all dedicated super-leaders at that level send their transaction counts to $s_z$. Consequently, $s_z$ sums all the counts of the number of transactions it received and sends the calculated sum to all the dedicated super-leaders of that level. 

\vspace{1mm}
{\bf Phase 3:}
In this phase, the algorithm determines the schedule to move the object $o$ to each dedicated super-leader and executes the transactions. This can be done in two ways: by using a universal TSP tour or an MST-based tour. 
Pseudocode and descriptions of the approaches are omitted due to space constraints.

\paragraph{\bf Analysis.}
In the fully distributed algorithm {\sc FullyDistributed\_SingleObj},
each node initially knows only its own transactions. The key challenge is to compute the set of super-leaders, which is handled in Phase 1 of the algorithm. After that, the algorithm proceeds similarly to Algorithm~\ref{alg:offline-single-obj}. Thus, any additional communication cost comes only from Phase 1 of the distributed algorithm.
Nevertheless, we show that {\sc FullyDistributed\_SingleObj} also achieves the same approximation as of {\sc GlobalAware\_SingleObj} (i.e., Algorithm~\ref{alg:offline-single-obj}).

Let $C$ be the total communication cost of Algorithm~\ref{alg:offline-single-obj} that involves the movement of transactions from current node to the respective dedicated super-leaders (say $C_1$) and the movement of object between the dedicated super-leaders in $H$ (say $C_2$). That means, $C=C_1+C_2$.

Let $C'$ be the total communication cost of {\sc FullyDistributed\_SingleObj},
which involves three phases. Let $C'_{p_1}, C'_{p_2},$ and $C'_{p_3}$ be the communication cost of Phase 1, Phase 2, and Phase 3, respectively. Then, $C' = C'_{p_1}+ C'_{p_2}+C'_{p_3}$. Since the execution of Phase 2 and Phase 3 of the algorithm is
equivalent to Algorithm~\ref{alg:offline-single-obj}, we have, $C'_{p_2}+C'_{p_3} = C$. That means, $C' = C'_{p_1}+C$.

Phase 1 
runs only once throughout the execution. Three types of messages are exchanged between different nodes of $H$ during the execution of Phase 1. 
\begin{itemize}
    \item [i.] Each node sends its transaction information to its cluster leader at level $l=0$ in $H$. This cost is at most $C_1$.
    \item [ii.] Each cluster leader at level $l\geq 0$ sends the count of total transactions in its cluster to the leader of parent cluster at level $l+1$. This cost is at most $C_2$.
    \item [iii.] When a cluster leader becomes a super-leader, it recursively notifies all child leaders down to the level $0$. This cost is also at most $C_2$. 
\end{itemize}

Hence, $C'_{p_1}\leq C_1 + 2C_2 < 2C$, which implies, $C'<2C+C < 3C$. Therefore, the following theorem is immediate:

\begin{theorem}
\label{theorem:fully-distributed-single-obj}
Algorithm~\ref{alg:offline-multiple-obj-alg} achieves $O(\log n \cdot \log D)$--approximation in communication cost when using universal TSP tour, and $O(\log D)$--approximation when using universal TSP tour of the object.
\end{theorem}

\subsection{Multiple Objects}
\label{sec:multi-object-fully-distributed}
In this section, we briefly describe how the fully distributed algorithm for single object can be extended to handle the transactions accessing multiple shared objects. 
Consider a set of shared objects $\cO =\{o_1, o_2, \ldots\}$, and 
each transaction $T_i \in \T$ accesses at most $k$ objects from $\cO$. 
As in {\sc FullyDistributed\_SingleObj}, transaction information is first shared through the levels of the hierarchy $H$.
Then, we calculate the dedicated super-leaders with respect to each object by running {\sc FullyDistributed\_SingleObj}. After that, the set of dedicated super-leaders for each transaction $T_i \in \T$ are determined by combining the object-specific super-leaders as in Algorithm~\ref{alg:offline-multiple-obj-alg}. 
Finally, we calculate the TSP tour on the dedicated super-leaders. Each object $o_j\in \cO$ visits the respective dedicated super-leader following the TSP tour and transactions are executed at each dedicated super-leader when the required objects gather.
 

%% file: conclusion.tex
\section{Conclusions}
\label{sec:conclusion}
In this paper, we studied transaction scheduling problem in distributed systems modeled as constant doubling dimension graphs. We used a dual-flow transaction execution model in which both objects and transactions can move within the network to minimize communication cost. 
We proposed centralized and fully distributed versions of the algorithm for transaction scheduling that achieve polylogarithmic approximations in communication cost.
In the future, it would be interesting to extend the algorithms and analysis for general and arbitrary graphs. Experimental evaluation of the designed algorithms against different application benchmarks in a practical setting would also be interesting.

%% file: references.bib
@inproceedings{busch2017fast,
author={Busch, Costas and Herlihy, Maurice and Popovic, Miroslav and Sharma, Gokarna},
  title={Fast scheduling in distributed transactional memory},
  booktitle={Proceedings of the 29th ACM Symposium on Parallelism in Algorithms and Architectures},
  pages={173--182},
  year={2017}
}

@inproceedings{Jia2005,
  title={Universal approximations for TSP, Steiner tree, and set cover},
  author={Jia, Lujun and Lin, Guolong and Noubir, Guevara and Rajaraman, Rajmohan and Sundaram, Ravi},
  booktitle={Proceedings of the thirty-seventh annual ACM symposium on Theory of computing},
  pages={386--395},
  year={2005}
}

@INPROCEEDINGS{Hendler2013,
  author = {Danny Hendler and Alex Naiman and Sebastiano Peluso and Francesco
	Quaglia and Paolo Romano and Adi Suissa},
  title = {Exploiting Locality in Lease-Based Replicated Transactional Memory
	via Task Migration},
  booktitle = {DISC},
  year = {2013},
  pages = {121-133}
}

@INPROCEEDINGS{Herlihy1993,
  author = {Herlihy, Maurice and Moss, J. Eliot B.},
  title = {Transactional Memory: Architectural Support for Lock-free Data Structures},
  booktitle = {ISCA},
  year = {1993},
  pages = {289--300},
  isbn = {0-8186-3810-9},
  location = {San Diego, California, USA}
}

@inproceedings{Kim2010,
  title={On transactional scheduling in distributed transactional memory systems},
  author={Kim, Junwhan and Ravindran, Binoy},
  booktitle={Symposium on Self-Stabilizing Systems},
  pages={347--361},
  year={2010},
  organization={Springer}
}

@Article{Sharma2014,
author="Sharma, Gokarna
and Busch, Costas",
title="Distributed transactional memory for general networks",
journal="Distributed Computing",
year="2014",
month="Oct",
day="01",
volume="27",
number="5",
pages="329--362",
issn="1432-0452",
doi="10.1007/s00446-014-0214-7"
}

@inproceedings{graphtm,
  author    = {Pavan Poudel and
               Gokarna Sharma},
  title     = {GraphTM: An Efficient Framework for Supporting Transactional Memory in a Distributed Environment},
  booktitle = {ICDCN},
  pages     = {11:1--11:10},
  year = {2020},
}

@article{nikoui2020cost,
  title={Cost-Aware Task Scheduling in Fog-Cloud Environment},
  author={Tina Samizadeh Nikoui and Ali Balador and Amir masoud Rahmani and Zeinab Bakhshi},
  journal={2020 CSI/CPSSI International Symposium on Real-Time and Embedded Systems and Technologies (RTEST)},
  year={2020},
  pages={1-8},
  publisher="IEEE",
}

@article{khiat2024genetic,
  title={Genetic-Based Algorithm for Task Scheduling in Fog--Cloud Environment},
  author={Khiat, Abdelhamid and Haddadi, Mohamed and Bahnes, Nacera},
  journal={Journal of Network and Systems Management},
  volume={32},
  number={1},
  pages={3},
  year={2024},
  publisher={Springer}
}

@article{mokni2023multi,
  title={Multi-objective fuzzy approach to scheduling and offloading workflow tasks in Fog--Cloud computing},
  author={Mokni, Marwa and Yassa, Sonia and Hajlaoui, Jalel Eddine and Omri, Mohamed Nazih and Chelouah, Rachid},
  journal={Simulation Modelling Practice and Theory},
  volume={123},
  pages={102687},
  year={2023},
  publisher={Elsevier}
}

@article{peixoto2021hierarchical,
  title={Hierarchical scheduling mechanisms in multi-level fog computing},
  author={Peixoto, Maycon Leone Maciel and Genez, Thiago AL and Bittencourt, Luiz F},
  journal={IEEE Transactions on services computing},
  volume={15},
  number={5},
  pages={2824--2837},
  year={2021},
  publisher={IEEE}
}

@article{busch2022dynamic,
  title={Dynamic scheduling in distributed transactional memory},
  author={Busch, Costas and Herlihy, Maurice and Popovic, Miroslav and Sharma, Gokarna},
  journal={Distributed Computing},
  volume={35},
  number={1},
  pages={19--36},
  year={2022},
  publisher={Springer}
}

@article{busch2023flexible,
  title={Flexible scheduling of transactional memory on trees},
  author={Busch, Costas and Chlebus, Bogdan S and Herlihy, Maurice and Popovic, Miroslav and Poudel, Pavan and Sharma, Gokarna},
  journal={Theoretical Computer Science},
  volume={978},
  pages={114184},
  year={2023},
  publisher={Elsevier}
}

@InProceedings{tilevich2002j,
author="Tilevich, Eli
and Smaragdakis, Yannis",
editor="Magnusson, Boris",
title="J-Orchestra: Automatic Java Application Partitioning",
booktitle="ECOOP 2002 --- Object-Oriented Programming",
year="2002",
publisher="Springer Berlin Heidelberg",
address="Berlin, Heidelberg",
pages="178--204",
isbn="978-3-540-47993-2"
}

@article{herlihy2007distributed,
  title={Distributed transactional memory for metric-space networks},
  author={Herlihy, Maurice and Sun, Ye},
  journal={Distributed Computing},
  volume={20},
  pages={195--208},
  year={2007},
  publisher={Springer}
}

@InProceedings{saad2011snake,
author="Saad, Mohamed M.
and Ravindran, Binoy",
editor="D{\'e}fago, Xavier
and Petit, Franck
and Villain, Vincent",
title="Snake: Control Flow Distributed Software Transactional Memory",
booktitle="Stabilization, Safety, and Security of Distributed Systems",
year="2011",
publisher="Springer Berlin Heidelberg",
address="Berlin, Heidelberg",
pages="238--252",
isbn="978-3-642-24550-3"
}

@article{srinivasagopalan2011oblivious,
  title={An oblivious spanning tree for single-sink buy-at-bulk in low doubling-dimension graphs},
  author={Srinivasagopalan, Srivathsan and Busch, Costas and Iyengar, SS},
  journal={IEEE Transactions on Computers},
  volume={61},
  number={5},
  pages={700--712},
  year={2011},
  publisher={IEEE}
}

@article{cao2023transaction,
  title={Transaction Scheduling: From Conflicts to Runtime Conflicts},
  author={Cao, Yang and Fan, Wenfei and Ou, Weijie and Xie, Rui and Zhao, Wenyue},
  journal={Proceedings of the ACM on Management of Data},
  volume={1},
  number={1},
  pages={1--26},
  year={2023},
  publisher={ACM New York, NY, USA}
}

@article{nikoui2016providing,
  title={Providing a cloud broker-based approach to improve the energy consumption and achieve a green cloud computing},
  author={Nikoui, Tina Samizadeh and Jabbehdari, Sam and Bagheri, Alireza},
  journal={International Journal of Computer Applications},
  volume={138},
  number={1},
  pages={42-49},
  year={2016},
  publisher={Foundation of Computer Science}
}

@inproceedings{shavit1995software,
  title={Software transactional memory},
  author={Shavit, Nir and Touitou, Dan},
  booktitle={Proceedings of the fourteenth annual ACM symposium on Principles of distributed computing},
  pages={204--213},
  year={1995}
}

@article{zhang2023efficient,
  title={Efficient Distributed Transaction Processing in Heterogeneous Networks},
  author={Zhang, Qian and Li, Jingyao and Zhao, Hongyao and Xu, Quanqing and Lu, Wei and Xiao, Jinliang and Han, Fusheng and Yang, Chuanhui and Du, Xiaoyong},
  journal={Proceedings of the VLDB Endowment},
  volume={16},
  number={6},
  pages={1372--1385},
  year={2023},
  publisher={VLDB Endowment}
}

@article{tran2023disco,
  title={DISCO: Distributed computation offloading framework for fog computing networks},
  author={Tran-Dang, Hoa and Kim, Dong-Seong},
  journal={Journal of Communications and Networks},
  volume={25},
  number={1},
  pages={121--131},
  year={2023},
  publisher={KICS}
}

@article{usui2010adaptive,
  title={Adaptive locks: Combining transactions and locks for efficient concurrency},
  author={Usui, Takayuki and Behrends, Reimer and Evans, Jacob and Smaragdakis, Yannis},
  journal={Journal of Parallel and Distributed Computing},
  volume={70},
  number={10},
  pages={1009--1023},
  year={2010},
  publisher={Elsevier}
}

@inproceedings{busch2015impossibility,
  title={Impossibility results for distributed transactional memory},
  author={Busch, Costas and Herlihy, Maurice and Popovic, Miroslav and Sharma, Gokarna},
  booktitle={Proceedings of the 2015 ACM Symposium on Principles of Distributed Computing},
  pages={207--215},
  year={2015}
}

@article{sharma2015load,
  title={A load balanced directory for distributed shared memory objects},
  author={Sharma, Gokarna and Busch, Costas},
  journal={Journal of Parallel and Distributed Computing},
  volume={78},
  pages={6--24},
  year={2015},
  publisher={Elsevier}
}

@inproceedings{abraham2006routing,
  title={Routing in networks with low doubling dimension},
  author={Abraham, Ittai and Gavoille, Cyril and Goldberg, Andrew V and Malkhi, Dahlia},
  booktitle={26th IEEE International Conference on Distributed Computing Systems (ICDCS'06)},
  pages={75--75},
  year={2006},
  organization={IEEE}
}

@inproceedings{busch2023stable,
  title={Stable Scheduling in Transactional Memory},
  author={Busch, Costas and Chlebus, Bogdan S and Kowalski, Dariusz R and Poudel, Pavan},
  booktitle={International Conference on Algorithms and Complexity},
  pages={172--186},
  year={2023},
  organization={Springer}
}

@article{PoudelRG24,
  author       = {Pavan Poudel and
                  Shishir Rai and
                  Swapnil Guragain},
  title        = {Ordered scheduling in control-flow distributed transactional memory},
  journal      = {Theor. Comput. Sci.},
  volume       = {993},
  pages        = {114463},
  year         = {2024},
  doi          = {10.1016/J.TCS.2024.114463},
  bibsource    = {dblp computer science bibliography, https://dblp.org}
}

@inproceedings{PoudelRS21,
  author       = {Pavan Poudel and
                  Shishir Rai and
                  Gokarna Sharma},
  title        = {Processing Distributed Transactions in a Predefined Order},
  booktitle    = {{ICDCN} '21: International Conference on Distributed Computing and
                  Networking, Virtual Event, Nara, Japan, January 5-8, 2021},
  pages        = {215--224},
  publisher    = {{ACM}},
  year         = {2021},
  doi          = {10.1145/3427796.3427819},
  bibsource    = {dblp computer science bibliography, https://dblp.org}
}

@inproceedings{adhikari2024spaastable,
  title={Stable Blockchain Sharding under Adversarial Transaction Generation},
  author={Adhikari, Ramesh and Busch, Costas and Kowalski, Dariusz R},
  booktitle={Proceedings of the 36th ACM Symposium on Parallelism in Algorithms and Architectures},
  pages={451--461},
  year={2024}
}

@inproceedings{kuhn2005locality,
  title={On the locality of bounded growth},
  author={Kuhn, Fabian and Moscibroda, Thomas and Wattenhofer, Rogert},
  booktitle={Proceedings of the twenty-fourth annual ACM symposium on Principles of distributed computing},
  pages={60--68},
  year={2005}
}

@INPROCEEDINGS{geo2009distributed,
  author={Gao, Jie and Guibas, Leonidas and Milosavljevic, Nikola and Dengpan Zhou},
  booktitle={2009 International Conference on Information Processing in Sensor Networks}, 
  title={Distributed resource management and matching in sensor networks}, 
  year={2009},
  volume={},
  number={},
  pages={97-108},
  doi={}}

@inproceedings{konjevod2008dynamic,
  title={Dynamic routing and location services in metrics of low doubling dimension},
  author={Konjevod, Goran and Richa, Andr{\'e}a W and Xia, Donglin},
  booktitle={Proceedings of the twenty-seventh ACM symposium on Principles of distributed computing},
  pages={417--417},
  year={2008}
}

@article{chan2016hierarchical,
  title={On hierarchical routing in doubling metrics},
  author={Chan, T-H Hubert and Gupta, Anupam and Maggs, Bruce M and Zhou, Shuheng},
  journal={ACM Transactions on Algorithms (TALG)},
  volume={12},
  number={4},
  pages={1--22},
  year={2016},
  publisher={ACM New York, NY, USA}
}

@article{JayaprakashS23,
  title={Approximation schemes for capacitated vehicle routing on graphs of bounded treewidth, bounded doubling, or highway dimension},
  author={Jayaprakash, Aditya and Salavatipour, Mohammad R},
  journal={ACM Transactions on Algorithms},
  volume={19},
  number={2},
  pages={1--36},
  year={2023},
  publisher={ACM New York, NY}
}

@article{NarmanHAS17,
  title={Scheduling internet of things applications in cloud computing},
  author={Narman, Husnu S and Hossain, Md Shohrab and Atiquzzaman, Mohammed and Shen, Haiying},
  journal={Annals of Telecommunications},
  volume={72},
  number={1},
  pages={79--93},
  year={2017},
  publisher={Springer}
}

@article{GoudarziPB23,
  title={Scheduling IoT applications in edge and fog computing environments: A taxonomy and future directions},
  author={Goudarzi, Mohammad and Palaniswami, Marimuthu and Buyya, Rajkumar},
  journal={ACM Computing Surveys},
  volume={55},
  number={7},
  pages={1--41},
  year={2022},
  publisher={ACM New York, NY}
}

@inproceedings{GramoliLTZ24,
  title={AOAB: optimal and fair ordering of financial transactions},
  author={Gramoli, Vincent and Lu, Zhenliang and Tang, Qiang and Zarbafian, Pouriya},
  booktitle={2024 54th Annual IEEE/IFIP International Conference on Dependable Systems and Networks (DSN)},
  pages={377--388},
  year={2024},
  organization={IEEE}
}

@article{KANG2022129,
  title={Job scheduling for big data analytical applications in clouds: A taxonomy study},
  author={Kang, Youyou and Pan, Li and Liu, Shijun},
  journal={Future Generation Computer Systems},
  volume={135},
  pages={129--145},
  year={2022},
  publisher={Elsevier}
}

@inproceedings{MulzerW20,
  title={Compact routing in unit disk graphs},
  author={Mulzer, Wolfgang and Willert, Max},
  booktitle={31st International Symposium on Algorithms and Computation (ISAAC 2020)},
  pages={16--1},
  year={2020},
  organization={Schloss Dagstuhl--Leibniz-Zentrum f{\"u}r Informatik}
}

@inproceedings{adhikari2023lockless,
  title={Lockless blockchain sharding with multiversion control},
  author={Adhikari, Ramesh and Busch, Costas},
  booktitle={International Colloquium on Structural Information and Communication Complexity},
  pages={112--131},
  year={2023},
  organization={Springer}
}

@inproceedings{GuptaKL03,
  author       = {Anupam Gupta and
                  Robert Krauthgamer and
                  James R. Lee},
  title        = {Bounded Geometries, Fractals, and Low-Distortion Embeddings},
  booktitle    = {44th Symposium on Foundations of Computer Science, {FOCS} 2003, Cambridge,
                  MA, USA, October 11-14, 2003, Proceedings},
  pages        = {534--543},
  publisher    = {{IEEE} Computer Society},
  year         = {2003},
}

@article{CeccarelloPPU17,
  title={MapReduce and streaming algorithms for diversity maximization in metric spaces of bounded doubling dimension},
  author={Ceccarello, Matteo and Pietracaprina, Andrea and Pucci, Geppino and Upfal, Eli},
  journal={Proceedings of the VLDB Endowment},
  volume={10},
  number={5},
  pages={469--480},
  year={2017},
  publisher={VLDB Endowment}
}

@article{KitamuraKOI21,
  title={Low-congestion shortcut and graph parameters},
  author={Kitamura, Naoki and Kitagawa, Hirotaka and Otachi, Yota and Izumi, Taisuke},
  journal={Distributed Computing},
  volume={34},
  number={5},
  pages={349--365},
  year={2021},
  publisher={Springer}
}
